\documentclass[11pt]{article}
\usepackage{amssymb, latexsym, mathtools, amsthm}
\begingroup
    \makeatletter
    \@for\theoremstyle:=definition,remark,plain\do{%
        \expandafter\g@addto@macro\csname th@\theoremstyle\endcsname{%
            \addtolength\thm@preskip\parskip
            }%
        }
\endgroup
\usepackage{enumerate}
\usepackage{pdfsync}
\usepackage{txfonts}
\usepackage[T1]{fontenc}
\usepackage{booktabs}
 \usepackage{graphicx}
\usepackage{xcolor}
\definecolor{dnrbl}{rgb}{0,0,0.3}
\definecolor{dnrgr}{rgb}{0,0.3,0}
\definecolor{dnrre}{rgb}{0.5,0,0}
\usepackage[colorlinks=true, citecolor=dnrgr, linkcolor=dnrre, urlcolor=dnrre] {hyperref}
\usepackage{parskip}
\usepackage{geometry}
\theoremstyle{plain}
\newtheorem{thm}{Theorem}[section]

\newtheorem{lem}[thm]{Lemma}

\newtheorem{defi}[thm]{Definition}

\newcommand{\Nat}{\mathbb{N}}

\newcommand{\restr}{\upharpoonright}  
\newcommand{\de}{\downarrow} 

\DeclarePairedDelimiter{\tuple}{\langle}{\rangle}

\newcommand{\bigo}[1]{\mathop{\bf O}\/\left({#1}\right)}
\newcommand{\smo}[1]{\mathop{\bf o}\/\left({#1}\right)}

\newcommand{\ml}{Martin-L\"{o}f }

\newcommand{\eg}{e.g.\ }
\newcommand{\ie}{i.e.\ }
\newcommand{\ce}{c.e.\ }

\newcommand{\lce}{left-c.e.\ }

\newcommand{\pf}{prefix-free }

\renewenvironment{abstract}
 { \normalsize
  \list{}{
    \setlength{\leftmargin}{.0cm}%
    \setlength{\rightmargin}{\leftmargin}%
    }%
  \item {\bf \abstractname.} \relax}
 {\endlist}

\newcommand{\pfm}{prefix-free machine }

\newcommand{\KG}{Ku\v{c}era-G\'{a}cs }

\newcommand{\pfn}{prefix-free}

\newcommand{\sz}{$\Sigma^0_1$ }

\newcommand{\CDFT}{Chalcraft, Dougherty, Freiling, and Teutsch }
\newcommand{\BDM}{Barmpalias, Downey, and McInerney }

\title{Lower bounds on the redundancy in computations from random oracles
via betting strategies with restricted wagers\thanks{Barmpalias was supported by the 
1000 Talents Program for Young Scholars from the Chinese Government, grant no.\ D1101130.
Additional support was received by
the Chinese Academy of Sciences (CAS) and the Institute of Software of the CAS.
Lewis-Pye was supported by a Royal Society University 
Research Fellowship. We thank the two anonymous referees and Peter G\'{a}cs for feedback that improved the
presentation of this article.}}
\author{George Barmpalias  \and Andrew Lewis-Pye \and Jason Teutsch}
\date{\today}
\begin{document}
\maketitle
\begin{abstract}
The \KG theorem \cite{MR820784,MR859105}
is a landmark result in algorithmic randomness asserting that every real is computable from a \ml random real. 
If the computation of the first $n$ bits of a sequence requires $n+h(n)$ bits 
of the random oracle, then $h$ is the {\em redundancy} of the computation.
Ku\v{c}era implicitly achieved redundancy $n\log n$ while G{\'a}cs  used a more elaborate coding procedure
which achieves redundancy $\sqrt{n}\log n$. A similar bound is implicit in the later proof
by Merkle and Mihailovi\'{c}
\cite{jsyml/MerkleM04}.
In this paper we obtain optimal strict lower bounds on the redundancy in computations from \ml random oracles.
We show that any nondecreasing  computable function $g$ such that 
$\sum_n 2^{-g(n)}=\infty$ is not a general upper bound 
on the redundancy in computations from \ml random oracles.
In fact, there exists a real $X$ such that the redundancy $g$ of any computation of $X$
from a \ml random oracle satisfies
$\sum_n 2^{-g(n)}<\infty$.
Moreover, the class of such reals is comeager and includes 
a $\Delta^0_2$ real as well as all weakly 2-generic reals. 
On the other hand, it has been recently shown in \cite{optcod16} 
that any real is computable from a \ml random oracle
with redundancy $g$, provided that $g$ is a computable nondecreasing function 
such that $\sum_n 2^{-g(n)}<\infty$. Hence our lower bound is optimal, and 
excludes many slow growing functions such as $\log n$ from bounding the redundancy in 
computations from random oracles for a large class of reals.
Our results are obtained as an application of a theory of effective
betting strategies with restricted wagers which we develop.
\end{abstract}
\noindent{\bf George Barmpalias}\\
\noindent State Key Lab of Computer Science, 
Institute of Software, Chinese Academy of Sciences, Beijing, China.
School of Mathematics, Statistics and Operations Research,
Victoria University of Wellington, New Zealand. 
\textit{E-mail:} \texttt{\textcolor{dnrgr}{barmpalias@gmail.com}}.
\textit{Web:} \texttt{\href{http://barmpalias.net}{http://barmpalias.net}}
\vfill
\noindent{\bf Andrew Lewis-Pye}\\
\noindent Department of Mathematics,
Columbia House, London School of Economics, 
Houghton St., London, WC2A 2AE, United Kingdom. 
\textit{E-mail:} \texttt{\textcolor{dnrgr}{A.Lewis7@lse.ac.uk.}}
\textit{Web:} \texttt{\textcolor{dnrre}{http://aemlewis.co.uk}} 
\vfill
\noindent{\bf Jason Teutsch}\\
\noindent Department of Computer and Information Sciences,
University of Alabama at Birmingham.\\
\textit{E-mail:} \texttt{\textcolor{dnrgr}{teutsch@uab.edu.}}
\textit{Web:} \texttt{\href{http://people.cs.uchicago.edu/~teutsch/}{http://people.cs.uchicago.edu/$\sim$teutsch}}
 \vfill\thispagestyle{empty}
\clearpage

\section{Introduction}
Every sequence is computable from a sequence which is random in the sense of \ml \cite{MR0223179}.
This major result in algorithmic information theory is known as the \KG theorem and was proved by
Ku\v{c}era \cite{MR820784, Kucera:87}
and G{\'a}cs \cite{MR859105}.
Both authors showed that the use of the oracle in these reductions can be
 bounded above by a computable function, but
Ku\v{c}era did not focus on minimizing the number of bits of the oracle that are needed to compute
the first $n$ bits of the sequence. If the latter number is $n+h(n)$, we say that the 
computation has {\em redundancy} $h$.
A close look at Ku\v{c}era's 
argument shows that his techniques achieve redundancy $n\log n$.
G{\'a}cs, on the other hand, took special care to minimise the oracle use. 
His argument produces a slightly more elaborate computation with redundancy 
$3\sqrt{n}\log n$, which can easily be improved to
$\sqrt{n}\log n$.
Both  of the arguments were formulated in terms of effective measure, \ie
according to the \ml definition of randomness.

The major difference between the results of 
Ku\v{c}era and G{\'a}cs is that the latter provides a reduction with oracle  use
$n+\smo{n}$ while the former does not. Merkle and Mihailovi\'{c}
\cite{jsyml/MerkleM04} presented a proof in terms of effective martingales, using similar ideas to
G{\'a}cs' proof but expressed in terms of betting strategies.
Up to now, the only known strict lower bound on the redundancy in computation from \ml random 
reals is the constant bound, and is due to
Downey and Hirschfeldt \cite[Theorem 9.13.2]{rodenisbook}.
Turing reductions
with constant redundancy are also known as {\em computably Lipschitz} or $cl$ reductions
and are well studied in computability theory, \eg see \cite[Chapter 9]{rodenisbook}.
Downey and Hirschfeldt  showed that
the redundancy in the \KG theorem cannot be $\bigo{1}$.
In fact, they constructed a sequence which is not computed with constant redundancy by any real whose Kolmogorov complexity
is bounded below by a computable nondecreasing unbounded function.
The reals with the latter property are sometimes known as {\em complex reals}. A close look at
this argument reveals that the set of reals which cannot be computed from any complex real with constant
redundancy is comeager. Moreover, it follows from the effective nature of the argument that: 
\begin{equation*}
\parbox{14.5cm}{\em a weakly 2-generic real cannot be computed by any complex real with
constant redundancy,}
\end{equation*}
where a real is called {\em weakly 2-generic} if it has a prefix in
every dense $\Sigma^0_2$ set of strings.

By \cite{BV2010} a real which is not complex has infinitely many initial segments of trivial complexity
in the sense that $C(X\restr_n)=C(n)+\bigo{1}$ and 
$K(X\restr_n)=K(n)+\bigo{1}$, where $K$ and $C$ denote the \pf and plain Kolmogorov complexities.
Sequences with the latter property are known as {\em infinitely often $C$-trivial and $K$-trivial} respectively.
It follows that any sequence computing a weakly 2-generic sequence
with constant redundancy is infinitely often $C$-trivial and infinitely often $K$-trivial.

\subsection{Our results, in context}
In Section \ref{V1dKIq77PY} we show that 
the redundancy in computations from \ml random oracles cannot be bounded by
certain slow growing functions. Recall that a real is $\Delta^0_2$
if and only if it is computable from the halting problem.

\begin{thm}\label{tbG7BLsZ}
There exists a real $X$ such that  $\sum_i 2^{-g(i)}<\infty$ for every nondecreasing 
computable function $g$ for which there exists a 
 \ml random real $Y$ which computes $X$ with redundancy $g$. 
  In fact, the reals $X$ 
 with this property form a comeager class which
includes every weakly 2-generic real.
\end{thm}

This result implies that any nondecreasing computable function $g$ such that 
$\sum_i 2^{-g(i)}=\infty$ is not a general upper bound on the redundancy in computations
of reals from \ml random oracles. A typical function with this property is $\lceil \log n \rceil$, so
the \KG theorem does not hold with redundancy $\lceil \log n \rceil$.
On the other hand, if $g(n)=2\cdot \lceil \log n \rceil$ then $\sum_i 2^{-g(i)}<\infty$.
It was recently shown in \cite{optcod16} that any nondecreasing computable function $g$
with the latter property is a general upper bound  on the redundancy in computations
of reals from \ml random oracles. Hence Theorem \ref{tbG7BLsZ} is optimal and gives a characterization
of the computable nondecreasing redundancy upper bounds
 in computations
of reals from \ml random oracles. Note that the optimal bounds obtained in  \cite{optcod16} 
are exponentially smaller than the
previously best known upper bound  of  $\sqrt{n}\log n$ from 
G\'{a}cs \cite{MR859105}.

With slightly more effort we also obtain an effective version of 
Theorem \ref{tbG7BLsZ}, which gives many more examples of reals $X$
which can only be computed from random oracles with large redundancy.
Recall that the halting problem relative to $A$ is denoted 
$A'$. The generalized non-low$_2$ reals are an important and extensively studied class in the context of degree theory:  $A$ is generalized low$_2$ if $A''$ has the same Turing degree as 
$(A\oplus\emptyset')'$, and a
 set which is not  generalized low$_2$ is called
generalized non-low$_2$.  

\begin{thm}[Jump hierarchy]\label{cdUMjzgrN}
Every set which is generalized non-low$_2$ (including the halting problem) 
computes a real
 $X$ with the properties of Theorem \ref{tbG7BLsZ}.
\end{thm}
The proof of Theorem \ref{tbG7BLsZ} also gives a nonuniform version of the latter result, requiring a weaker
condition regarding the computational power of the oracle.
Recall from \cite{DJS2} that a set $A$ is array noncomputable if for each function $f$
that is computable
from the halting problem with a computable upper bound on the oracle use, there exists a function
$h$ which is computable from $A$ and which is not dominated by $f$. A degree is array noncomputable if its members are. 
The class of array noncomputable degrees (again an extensively studied class) is an upwards closed  superclass of the generalized non-low$_2$ degrees, and includes low degrees amongst its members. 

\begin{thm}[Array noncomputability]\label{2NY34dfIMG}
Suppose that $\sum_i 2^{-g(i)}=\infty$ for some
computable nondecreasing function $g$.
Then every array noncomputable real 
computes a real $X$ 
which is not computable by any \ml real with redundancy $g$.
\end{thm}

The proof of all of the above results relies on an analysis of effective betting strategies
with restricted wagers. This is not entirely surprising as (a) \ml randomness can be expressed 
in terms of the success of effective martingales (see Section \ref{BHdbZ8A36P})
and (b) there is a direct connection between Turing reductions, semi-measures and martingales,
which goes back to Levin and Zvonkin \cite{MR0307889} (see the discussion before 
Section \ref{RW7DicsV4p}). 
A strategy (or martingale) is said to have restricted wagers when it can only bet amounts
from a given set of possible values, where this set of legitimate values may be allowed to vary
from stage to stage of the betting game. 
The subject of martingales with restricted wagers has been the focus of intense research activity recently. The simplest case
is when the restriction specifies only a  minimum amount that the gambler can bet at each stage.
Integer-valued martingales are examples of strategies of this type, and were
motivated and studied
by  Bienvenu, Stephan and Teutsch \cite{cie/BienvenuST10, BSTmartin}, 
\CDFT \cite{TeutChalcraft}, Teutsch \cite{Teu14agCWGnp}, \BDM \cite{JCSSbarmp15}
and most recently Herbert \cite{herbertlow}.
A more general study of betting strategies with restricted wagers
can be found in Peretz \cite{Peretzwager} and Bavly and Peretz \cite{Peretzagainst}.
Given a function $g$, a function on binary strings is called {\em $g$-granular} if
its value on any string $\sigma$ is an integer multiple of $2^{-g(|\sigma|)}$.
The notion of $g$-granular supermartingales is based on the above notion, and 
is a formalisation of the intuitive notion of betting strategies with restricted wagers.
We defer the formal definition until
Section \ref{RBq4QNezFK}, but state the following pleasing result now, which indicates their importance. 
Let $\lambda$ denote the empty string. 
The definition of c.e.\ supermartingales and 
other basic terms will be reviewed in Section \ref{BHdbZ8A36P}. 

\begin{thm}[Granular supermartingales]\label{Wab5YAaQ7s}
Suppose that $g$  is a nondecreasing and computable function.
\begin{enumerate}[\hspace{0.5cm}(a)]
\item If $\sum_i 2^{-g(i)}<\infty$, for every \ce supermartingale $N$ there exists
a $g$-granular \ce supermartingale $M$ such that  for each  $X$ we have
$\limsup_s M(X\restr_n)=\infty$ if and only if $\limsup_s N(X\restr_n)=\infty$.
\item If $\sum_i 2^{-g(i)}=\infty$, there exists  a \ce supermartingale
$N$ such that for all $g$-granular \ce supermartingales $M$ there exists some $X$ such that
$\limsup_s N(X\restr_n)=\infty$ and $\limsup_s M(X\restr_n)<\infty$.
\end{enumerate}
\end{thm}
Informally,  the first clause of Theorem \ref{Wab5YAaQ7s} expresses the fact that
if $\sum_i 2^{-g(i)}<\infty$ for a computable nondecreasing function $g$, then
$g$-granular supermartingales suffice for the definition of \ml randomness.
The second clause of Theorem \ref{Wab5YAaQ7s} says that, in fact, 
$\sum_i 2^{-g(i)}<\infty$ is also a necessary condition for the sufficiency of 
$g$-granular supermartingales for the purpose of defining \ml randomness.
The proofs of Theorem \ref{tbG7BLsZ},
Theorem \ref{cdUMjzgrN} and Theorem \ref{2NY34dfIMG}
 rely on clause (b) of Theorem \ref{Wab5YAaQ7s},
and more specifically on the following more detailed version of this statement, 
which is of independent interest.

\begin{lem}\label{Bb4WohYKDQ}
Suppose that nondecreasing $g$ is computable and $\sum_i 2^{-g(i)}=\infty$. Given any $g$-granular
\ce supermartingale $M$ there exists a $(g+1)$-granular \ce supermartingale $N$ 
and a real $X$ which is computable from $M$, 
such that $\limsup_n M(X\restr_n)\leq M(\lambda)$ and 
$\limsup_n N(X\restr_n)=\infty$.
\end{lem}

Lemma \ref{Bb4WohYKDQ} clearly implies clause (b) of Theorem \ref{Wab5YAaQ7s}, since it implies that the universal c.e.\ supermartingale will satisfy  Theorem \ref{Wab5YAaQ7s} (b).
However it is stronger than the latter, because the supermartingale $N$ is said to be
$(g+1)$-granular, \ie just a single step more granular than the given supermartingale $M$.

\subsection{Further related work in the literature}
The present work is a step towards characterizing the optimal redundancy that 
can be achieved through a general  process for coding reals into  \ml random reals, which was completed
in \cite{optcod16}.
Doty \cite{cie/Doty06} revisited the \KG theorem from the viewpoint of constructive dimension.
He characterized the optimal asymptotic ratio between $n$ and the use on argument 
$n$ when a random oracle computes $X$,   in terms of the constructive dimension of $X$.
Recall that the effective packing dimension of a real can be defined as
\[
\textrm{Dim}(X)=\limsup_n\frac{K(X\restr_n)}{n}.
\]
Doty \cite{cie/Doty06} showed that the number of bits of a random
oracle needed to compute $X\restr_n$ is at most $\textrm{Dim}(X)\cdot n + \smo{n}$.
So for any real $X$
with $\textrm{Dim}(X) < 1$ , its redundancy is negative on almost all of its prefixes.
On the other hand, any \ml random real has redundancy 0 since it reduces to itself. 
Thus, Theorem \ref{tbG7BLsZ} refers to reals that are non-random, but that have effective packing dimension 1. 
One difference between Doty's work and our project is that we are looking to characterize
the redundancy that is possible {\em for every sequence} regardless its effective dimension.
A second difference with the work in \cite{cie/Doty06} (as well as \cite{jsyml/MerkleM04}) is that 
we are interested in precise bounds on the redundancy of computations from \ml random reals, rather than
just the asymptotic ratio between $n$ and the use on argument $n$. 

Asymptotic conditions on the redundancy $g$ in computations from random oracles
such as the ones in Theorem \ref{Wab5YAaQ7s}, have been used with respect to Chaitin's $\Omega$
in Tadaki \cite{Tadaki:1611072} and 
Barmpalias, Fang and Lewis-Pye \cite{asybound}. However the latter work only refers to computations of
computably enumerable sets and reals and does not have essential connections with the present work, except perhaps
for some apparent analogy of the statements proved.

\section{Betting strategies with restricted wagers}
The proof of our main result, Theorem \ref{tbG7BLsZ},
relies substantially on a lemma concerning effective betting strategies,
as formalised by martingales.
This section is devoted to proving that lemma, but is also a contribution to the study of
strategies with restricted wagers. We are interested in strategies where the wager at step $s$ of the game
must be an integer multiple of a rational number which is a function of $s$.
In the next subsection we summarise some required  background material.  

\subsection{Algorithmic randomness and effective martingales}\label{BHdbZ8A36P}
The three main approaches to the definition of algorithmically random sequences
are based on (a) incompressibility (Kolmogorov complexity), (b) unpredictability
(effective betting strategies) and (c) measure theory (effective statistical tests).
There are direct translations between any pair of (a), (b) or (c),
and most notions of algorithmic randomness (of various strengths) are naturally defined
via any  of these approaches. The first two approaches are most relevant to the present work.  

Informally, the Kolmogorov complexity of a string is the length of its shortest description.
The concept of {\em description} is formalised via the use of a Turing machine $V$. Given $V$, 
we say that $\sigma$ is a description of $\tau$ if $V(\sigma)$ is defined and equal to $\tau$.
There are different  versions of Kolmogorov complexity that may be considered, 
depending on the type of machine that is used in order to formalise the concept of a description.
Prefix-free complexity, based on  \pf machines,
is just one way to approach algorithmic randomness, and is the notion of complexity that we shall  use
in order to obtain our results here. Note, however, that our main results concern only the 
robust concept of \ml randomness, which can be defined equivalently with respect to a number of
different machine models (or more generally via a number of
diverse approaches, as we discuss in the following).
A set of binary strings is \pf if it does not contain any pair of distinct strings such that one is
an extension of the other.
A \pf Turing machine is a Turing machine with domain which is a \pf subset of the finite binary strings.
The \pf Kolmogorov complexity of a string $\sigma$ with respect to a \pf Turing machine $N$, 
denoted $K_N$, is the length of the shortest string $\tau$ such that $N(\tau)\de=\sigma$.
Let $(N_e)$ be an effective list of all \pf machines.
Prefix-free Kolmogorov complexity is based on the existence of
an optimal universal prefix-free machine  $U$ \ie such that $K_U$ is minimal, modulo a constant,
amongst all $K_{N_e}$. For the duration of this paper, we adopt a standard choice for $U$, which is
defined by 
$U(0^e\ast 1\ast\sigma)\simeq N_e(\sigma)$ (where `$\simeq$' means that
 one side  is defined iff  the other is, 
and that if defined  the two sides are equal). 
From this definition it follows immediately that  $K_U(\sigma)\leq K_{N_e}(\sigma)+e+1$ for all
$\sigma$ and all $e$.
Clearly $U$ is a universal  \pf machine which can simulate any other \pf machine
with only a constant overhead, the size of its index. For simplicity we let
$K(\sigma)$ denote $K_U(\sigma)$, \ie when the underlying \pf machine is the default $U$, we
suppress the subscript in the notation of Kolmogorov complexity.
We identify subsets of $\Nat$ with their characteristic functions, viewed as an infinite binary sequences, 
and often refer to them as 
{\em reals}. Given a real $A$, we let $A\restr_n$ denote the first $n$ bits of $A$. 
The algorithmic randomness of infinite binary sequences is often defined
in terms of \pf Kolmogorov complexity. 
We say that an infinite binary sequence $A$ is 1-random if  there exists a constant $c$
such that $K(A\restr_n)\geq n-c$ for all $n$. Informally, these are the infinite sequences 
for which all initial segments are incompressible. 

An equivalent definition of algorithmic randomness for reals
can be given in terms of effective statistical tests 
\cite{MR0223179}.
A \ml test is an effective sequence of \sz classes $(V_e)$ (which we may view as
a uniformly \ce sequence of sets of strings) such that $\mu(V_e)<2^{-e}$ for each $e$.
A real $X$ is \ml random if $X\notin \cap_e V_e$ for any \ml test $(V_e)$.
A third way to define
algorithmic randomness, due to Schnorr \cite{Schnorr:75,MR0354328}, 
can be given in terms of 
betting strategies, normally formalised as martingales or supermartingales.
We are interested in supermartingales as functions 
$h: 2^{<\omega}\to\mathbb{R}^{\geq 0}$
with the property $h(\sigma 0)+h(\sigma 1)\leq 2 h(\sigma)$. 
A supermartingale
such that $h(\sigma 0)+h(\sigma 1)= 2 h(\sigma)$ for all $\sigma$ is called a martingale.
We say that: 
\[
\textrm{the supermartingale $h$ {\em succeeds on a real} $X$ if $\limsup_s h(X\restr_n)=\infty$.}
\]
Note that a stronger notion of success is the condition $\lim_s h(X\restr_n)=\infty$.
In many situations, such as in the  characterization of \ml random sequences in terms of martingales
(see below), 
it is not important which notion of success is used. In the present work, however,  it
seems more appropriate to use the weaker notion as a default, and to mention the stronger notion 
explicitly when it plays a role in an argument.
We say that a function $f:2^{<\omega}\to\mathbb{R}^{\geq 0}$ is \lce if there is a computable function 
$f_0:2^{<\omega}\times\Nat\to\mathbb{Q}^{\geq 0}$ 
which is nondecreasing in the second argument and such that
$f(\sigma)=\lim_s f_0(\sigma,s)$ for each $\sigma$. 
In this case the function $f_0$ is called the \lce approximation to $f$.
A (super)martingale is called \ce 
if it is \lce as a function. It is a well known fact, due to Schnorr \cite{Schnorr:75,MR0354328}
(see for example \cite[Theorems 6.2.3, 6.3.4]{rodenisbook}), that the following are equivalent
for each real $X$:
\begin{enumerate}[\hspace{0.5cm}(i)]
\item $X$ is \ml random;
\item no \ce supermartingale succeeds on $X$;
\item $K(X\restr_n)\geq n-c$ for some constant $c$ and all $n$.
\end{enumerate}
In fact this equivalence is effective, 
in the sense of Lemma \ref{QhFm7xcjbM}.
Recall that $\lambda$ denotes the empty string. 
The weight of a \pf set of strings $S$ is  $\sum_{\sigma \in S} 2^{-|\sigma|}$,  and
is equal to the measure of the \sz class of reals represented by $S$, \ie the reals that have a prefix in $S$.

\begin{lem}[Schnorr, implicit in \cite{Schnorr:75,MR0354328}]\label{QhFm7xcjbM}
Given the index for a c.e.\ supermartingale $M$, $m\geq M(\lambda)$ and  $c\in \mathbb{N}$, one can effectively find $k$ for which the following holds: 
any real with a prefix $\sigma$ with $M(\sigma)\geq k$, has a prefix $\tau$ with
$K(\tau)\leq |\tau| -c$.
\end{lem}
\begin{proof}
Given a supermartingale $M$ with $M(\lambda)\leq m$, 
by Kolmogorov's inequality the measure of reals $X$ for which there exists $n$ such that
$M(X\restr_n)\geq k$ is bounded above by $m/k$.
On the other hand, given an integer $c$ and 
a \pf and  \ce set of finite strings $V$ 
such that the weight of $V$ is bounded above by $2^{-c}$, we can effectively
define a \pf machine $N$ such that $K_N(\sigma)\leq |\sigma|-c$ for all 
$\sigma\in V$ (this is a typical application of the so-called Kraft-Chaitin online algorithm for
the construction of a \pf machine). The crucial point is that
given a \ce supermartingale $M$ and $k$, the set of
reals $X$ such that $M(X\restr_n)\geq k$ for some $n$, is a \sz class. 
Hence for each $k$ we may effectively obtain a \pf set $V(k)$ of strings $\sigma$ 
of weight $\leq m\cdot k^{-1}$ such that every real $X$ with 
$M(X\restr_n)\geq k$  for some $n$, has a prefix in $V(k)$.

We can effectively find $k$ as required by the lemma, via the recursion theorem (and its uniformity) as follows.
We construct a \pf machine $N$, and by the recursion theorem we may use its index
$b$ in its definition. We let $k$ be $2^{m+b+c+1}$ and define $N$ via the Kraft-Chaitin
online algorithm such that 
$K_{N}(\sigma)\leq |\sigma|-c-b-1$  for all $\sigma\in V(k)$.
Since the measure of $V(k)$ is bounded above by
$m\cdot 2^{-c-m-b-1}<2^{-c-b-1}$, the definition of $N$ is valid, and the application of the
 Kraft-Chaitin
online algorithm along with the definition of the sets $V(k)$ ensures that 
$K_{N}(\sigma)\leq |\sigma|-c-b-1$  
for all  strings $\sigma$ in $V(k)$.
But according to our choice of optimal universal machine $U$
this implies that
$K(\sigma)\leq |\sigma|-c$  
for all  strings $\sigma$ in $V(k)$.
By the choice of $V(k)$, this means that every real $X$ such that 
for all  strings $\sigma$ such that $M(\sigma)\geq k$, 
$M(X\restr_n)\geq k$  for some $n$, has a prefix $\sigma$ with 
$K(\sigma)\leq |\sigma|-c$,
which concludes the proof.\footnote{The published version of this paper includes
a slightly stronger version of Lemma \ref{QhFm7xcjbM} which is not correct. However
it is the present version of  Lemma \ref{QhFm7xcjbM} that is and was actually used in Section \ref{q4yEFszH5d} so this is a minor correction.}
\end{proof}

Martingales are expressions of betting games on sequences of binary outcomes. 
More specifically, if $h$ is a martingale, then $h(\sigma)$ can be thought of as expressing the 
capital of the player betting according to the strategy $h$, after the sequence $\sigma$ of outcomes. If at
state $\sigma$ of the game we bet $\alpha$ on 0, then our capital at the next stage will be
$h(\sigma0)=h(\sigma)+\alpha$ or $h(\sigma1)=h(\sigma)-\alpha$
according to whether the outcome was 0 or 1, respectively.
So $h(\sigma0)+h(\sigma1)=2h(\sigma)$. Martingales can therefore be seen as 
modeling the capital in a betting game
along every possible sequence of outcomes. Given a martingale $h$, the amount
that is bet at state $\sigma$ is $|h(\sigma0)-h(\sigma)|=|h(\sigma1)-h(\sigma)|$ 
and is bet on 0 or 1 according to whether $h(\sigma0)>h(\sigma1)$ or not.
Hence every martingale determines a betting strategy, which we may regard as a function
from strings to the non-negative reals, which determines what amount is bet and on which outcome.
Conversely, a betting strategy corresponds to a martingale, which models the remaining capital
at the end of each bet. Our definition of granular betting strategies 
in Section \ref{RBq4QNezFK} relies on the condition
that the bets made at each stage (and not necessarily the remaining capital) are granular, in the
sense that they correspond to numbers from a specific set.
For more detailed background on the notions discussed in this section,  
we refer the reader to \cite[Chapter 6]{rodenisbook}. For a general introduction to
algorithmic randomness we refer to \cite{Li.Vitanyi:93}.

\subsection{Restricted martingales}\label{RBq4QNezFK}
Restricting  the set of possible betting strategies may give rise to weaker forms of randomness.
There are many ways to impose such restrictions, but the method which is relevant to our work
involves dictating a minimum wager at each step of the betting process, and requiring that
the gambler bets an integer multiple of that minimum wager. 
We formalise this notion in the following definitions.

\begin{defi}[Granularity of functions]
Given functions $g:\Nat\to\Nat$ and $M:2^{<\omega}\to\mathbb{R}$, 
we say that $M$ is $g$-granular (or has granularity $g$)
if for every string $\sigma$ the value of $M(\sigma)$ is an integer multiple of $2^{-g(|\sigma|)}$. 
\end{defi}
We could now restrict our attention to supermartingales that are $g$-granular as functions, for some
computable non-decreasing function $g$. Indeed, this approach suffices for
most of the results in this paper. However we formalise betting strategies with restricted wagers
in a slightly more general way, which is both intuitively justifiable and also allows to prove
the rather elegant characterization of Theorem
\ref{Wab5YAaQ7s}.
\begin{defi}[Granular \ce supermartingales]\label{tkUAG5IyAt}
Given a nondecreasing computable function $g:\Nat\to\Nat$, 
we say that a  \ce supermartingale $M$
is $g$-granular if there exists a computable sequence of rationals $(q_i)$  and
a $g$-granular \lce function $N:2^{<\omega}\to \mathbb{R}$
such that $M(\sigma)=N(\sigma)+ \sum_{i\geq |\sigma|}q_i$. In the special case where $f(n)=\sum_{i\geq n}q_i$
is constantly zero we say that $M$ is a strongly $g$-granular \ce supermartingale.
\end{defi}
Intuitively speaking, the function $f$ in the above definition represents a part of the capital
which is not used for betting, and is transferred from each round to the next round, perhaps 
reduced due to inflation, in accordance with the standard interpretation of supermartingales as betting
strategies. More precisely, the value of $f$ does not depend on the particular bets that we have placed
up to a certain stage, but rather on the number of these bets, \ie the stage of the game.
 The particular case where $f$ is the zero function is of special importance,
as it is the notion that will be used in the proofs of most of the results in this paper. We emphasize the fact
that in Definition \ref{tkUAG5IyAt} we require $N$ to be a \lce function, and so a \ce index of a 
granular \ce supermartingale $M$ is not merely a program which gives a \lce approximation to $M$ but
a program that enumerates the values  $(q_i)$ and also gives a \lce approximation $N$ -- thereby specifying a \lce approximation to $M$. 

We are ready to present and prove the main result of this section, 
which is a more elaborate version of Lemma \ref{Bb4WohYKDQ}.
Clearly Lemma \ref{U26qaoGM} implies Lemma \ref{Bb4WohYKDQ}.
 However Lemma \ref{U26qaoGM} also gives the rate of growth of
 the supermartingale $N$ as a function of $g$, which is absent in the statement of 
 Lemma \ref{Bb4WohYKDQ}. Summing up, 
Lemma \ref{U26qaoGM} implies Lemma \ref{Bb4WohYKDQ}, which in turn implies
clause (b) of Theorem \ref{Wab5YAaQ7s}. The reason we preceded the following elaborate statement 
with the two weaker ones, is that the additional technical information may only be of interest to some readers,
and may distract others from the main result, namely Theorem  \ref{Wab5YAaQ7s}.

\begin{lem}[Granular \ce supermartingales]\label{U26qaoGM}
Given a string $\nu_0$, a nondecreasing 
computable function $g:\Nat\to\Nat$  
and a $g$-granular \ce supermartingale $M$,
there exists a real $X\supset\nu_0$
and a $(g+1)$-granular \ce supermartingale $N$ 
such that $N(X\restr_n)\geq \sum_{0\leq i\leq n} 2^{-g(i)-1}$
and $M(X\restr_n)\leq M(\nu_0)$ for all $n\geq |\nu_0|$.
Moreover if $M$ is strongly $g$-granular, then $N$ can also be chosen to be strongly $(g+1)$-granular.
\end{lem}
\begin{proof}
For the sake of ease of notation, we may assume that $\nu_0$ is the empty string.
The proof of the more general case is a direct adaptation of the
proof of this special case.
 A first naive attempt would be to let $N$ bet in the opposite way to $M$, which
means to define
\[
N(\sigma \ast i) =
\begin{cases}
N(\sigma)+2^{-g(|\sigma|+1)} &\textrm{if $M(\sigma \ast i)<M(\sigma \ast (1-i))$}\\
N(\sigma)-2^{-g(|\sigma|+1)} &\textrm{otherwise}
\end{cases}
\]
and let $X$ carve a path on the binary tree where $N$ wins (so $M$ loses) at every stage (ignoring for now the possibility that $M(\sigma \ast 0)=M(\sigma \ast 1)$).
This martingale, however,  is not necessarily c.e.,  because $M$ is merely \ce and not
computable, so the condition $M(\sigma \ast i)<M(\sigma \ast (1-i))$
is not decidable. Following the same basic idea (letting $N$ bet on the outcomes where $M$ does not 
increase its capital)
we produce a more sophisticated definition, which defines  $N$ as a $(g+1)$-granular 
\ce supermartingale. 

The idea for this argument is to effectivize the above definition of $N$ so that
the resulting function is a \ce supermartingale. In order to do this, we need to avoid using the condition
$M(\sigma \ast i)<M(\sigma \ast (1-i))$ in the above definition of $N$, since it is not decidable.
The solution is to incorporate the effective approximations to $M(\sigma \ast i), M(\sigma \ast (1-i))$ into the definition of
$N(\sigma \ast i), N(\sigma \ast (1-i))$ in such a way that we can still gain additional capital by choosing the right value of $i$. 
It turns out that we can do this by using an additive term of $2^{-g(|\sigma|+1)-1}$, thus making $N$ a $(g+1)$-granular
supermartingale, as indicated in \eqref{Qd5AiJd}. In the following we formalize this idea, and prove that it works.

 Let $(M_s)=(\hat{M}_s+f_s)$ be a \lce approximation to $M$, such that each $M_s$ is a 
$g$-granular supermartingale, $f_s(n)=\sum_{i:n\leq i \leq s}q_i $ for the computable sequence of rationals $q_i$,   $\hat{M}_s(\sigma)$ is a $g$-granular function for each $s$ and such that $\hat{M}_s(\sigma)$ is nondecreasing as a function of $s$. 
The reader may find the proof more tractable if they assume $f$ to be constantly zero.
This corresponds to the case of the lemma regarding strongly granular supermartingales and
 contains all the important ideas of the general proof. For completeness, however,  
we present the full argument here.
There
exist \lce integer-valued functions $t:2^{<\omega}\to \Nat$,
$q: 2^{<\omega}\to \Nat$ with \lce approximations $t_{s}, q_s$ 
such that:
\begin{equation}\label{an8o7LvU}
\parbox{11cm}{$\hat{M}_s(\sigma\ast 0)=t_s(\sigma)\cdot 2^{-g(|\sigma|+1)}$ 
\hspace{0.5cm}and \hspace{0.5cm}
$\hat{M}_s(\sigma\ast 1)=q_s(\sigma)\cdot 2^{-g(|\sigma|+1)}$}
\end{equation}
for all $\sigma,s$. Recall that $\lambda$ denotes the empty string. 
We will define a computable sequence of supermartingales $(N_s)$, which is also
a \lce approximation to their limit $N$, a \ce supermartingale.
In fact, we will define a
computable sequence of $(g+1)$-granular  functions $(\hat{N}_s)$ such that 
the functions $N_s(\sigma):= \hat{N}_s(\sigma)+f_s(|\sigma|)$ are 
computable supermartingales. Then clearly the limit $N$ of $N_s$ will be a supermartingale
and  by Definition \ref{tkUAG5IyAt}, the function $N$ will also be a
$(g+1)$-granular supermartingale.
Let $\hat{N}_s(\lambda)= \hat{M}_s(\lambda) +2^{-g(0)-1}$ for all stages $s$ and let 
$\hat{N}_0(\sigma)=0$ for all nonempty strings $\sigma$.
The values of $\hat{N}_s(\sigma)$ for $s>0$ and nonempty strings $\sigma$ are defined inductively
as follows. 
We order the strings first by length and then lexicographically.
The notion of accessibility is defined dynamically during the construction.
At stage 0, no string has been accessed.

{\bf Construction of $\hat{N}_s$.}
At each stage $s+1$, if 
$\hat{M}_{s+1}(\lambda)\neq \hat{M}_{s}(\lambda)$ 
then do nothing other than define $\hat{N}_{s+1}(\lambda)=\hat{M}_{s+1}(\lambda) +2^{-g(0)-1}$ 
 and $\hat{N}_{s+1}(\tau)=\hat{N}_s(\tau)$ for all $\tau\neq\lambda$ ($g$-granularity 
 means this can only occur at finitely many stages).  
 Otherwise, find the least string $\sigma$ of length at most $s$, such that for all $\eta\subseteq\sigma$, 
\begin{equation}\label{ipmD1WlxGj}
\hat{N}_s(\eta)\geq \hat{M}_s(\eta)+\sum_{i\leq |\eta|} 2^{-g(i)-1}
\end{equation}
and one of the following clauses holds:
\begin{enumerate}[\hspace{0.5cm}(a)]
\item $\sigma$ has not been accessed at any stage $\leq s$;
\item $\sigma$ was last accessed at stage $m<s+1$ and either $t_{s+1}(\sigma)\neq t_m(\sigma)$
or $q_{s+1}(\sigma)\neq q_m(\sigma)$.
\end{enumerate}
If such a string does not exist, let $\hat{N}_{s+1}(\eta)=\hat{N}_{s+1}(\eta)$ for all strings $\eta$. Otherwise define:
\begin{equation}\label{Qd5AiJd}
\begin{cases}
\hat{N}_{s+1}(\sigma \ast 0) =&\sum_{i\leq |\sigma|} 2^{-g(i)-1}+q_{s+1}(\sigma)\cdot 
2^{-g(|\sigma|+1)} + 2^{-g(|\sigma|+1)-1}\\
\hat{N}_{s+1}(\sigma \ast 1) =&\sum_{i\leq |\sigma|} 2^{-g(i)-1}+t_{s+1}(\sigma)\cdot 
2^{-g(|\sigma|+1)} - 2^{-g(|\sigma|+1)-1}
\end{cases}
\end{equation}
and declare that  $\sigma$ has been  accessed at stage $s+1$.
Note that in this case we have $\hat{M}_s(\sigma)=\hat{M}_{s+1}(\sigma)$,
because if this was not true and $\eta$ is the immediate predecessor of $\sigma$, then
$t_{s+1}(\eta)\neq t_s(\eta)$ or $q_{s+1}(\eta)\neq q_s(\eta)$, which contradicts the minimality of $\sigma$.

For $\tau$ other than $\sigma \ast 0$ and $\sigma \ast 1$ define  $\hat{N}_{s+1}(\tau) =\hat{N}_{s}(\tau)$. 
 Note also that (a) the roles of $q$ and $t$ are reversed in the above definition in the sense that $q_{s+1}$ is used in the definition of $\hat{N}_{s+1}(\sigma \ast 0)$ rather than  $\hat{N}_{s+1}(\sigma \ast 1)$, and (b) the definitions of $\hat{N}_{s+1}(\sigma \ast 0)$ and $\hat{N}_{s+1}(\sigma \ast 1)$ are \emph{not} symmetrical, since we add $2^{-g(|\sigma|+1)-1}$ in defining the former value, while we subtract it in defining the latter. 

This concludes the construction of the functions $\hat{N}_s$ and we let 
$N_s(\sigma)=\hat{N}_s(\sigma)+f_s(|\sigma|)$ for all $\sigma$.
We also let $N(\sigma)=\lim_s N_s(\sigma)$ for all $\sigma$. 

{\bf Intuition for the construction.}
The driving force behind the construction is \eqref{ipmD1WlxGj}, which
is guaranteed to hold for the empty string, but not for all strings.
However, as we are going to verify in the following,
inductively we can argue that there is real $X$ such that
all of its initial segments $\eta$ satisfy \eqref{ipmD1WlxGj}.
The updates defined in \eqref{Qd5AiJd} ensure that \eqref{ipmD1WlxGj} continues to hold for at least
one immediate extension of $\sigma$.
The updates are made gradually, following the approximations to $M$, in order to ensure that $N$ is
a \ce supermartingale. Moreover the equations in the update mechanism \eqref{Qd5AiJd} will ensure that
$N$ is a $(g+1)$-granular supermartingale, as required.

{\bf Verification.}
 The fact that $g$ is nondecreasing means that $\hat{N}_t$ can never take negative values (in particular the term $- 2^{-g(|\sigma|+1)-1}$ in the definition of $\hat{N}_{t+1}(\sigma \ast 1)$ cannot cause negative values).  By \eqref{Qd5AiJd} and the fact that $t_s(\sigma),q_s(\sigma)$ are
 nondecreasing we have that 
 \begin{equation}\label{LjU5qbvX}
\hat{N}_t(\sigma)\leq \hat{N}_{t+1}(\sigma) 
\hspace{0.5cm}
\textrm{for all $t,\sigma$}.
\end{equation}
Since $N_s(\sigma)=\hat{N}_s(\sigma)+f_s(|\sigma|)$,
and $(f_s)$ is a \lce approximation to the function $f$, it follows that
$(N_s)$ is a \lce approximation to the limit $N$
of $(N_s)$. Hence $N$ is a \lce function. 

Next, 
we verify that each $N_t$ is a supermartingale. We must show that for all $t$:
\begin{equation}\label{jEIyglTt}
N_t(\sigma\ast 0)+N_t(\sigma\ast 1)\leq 2 \cdot N_t(\sigma).
\end{equation}
For $t=0$ this property clearly holds.  Given any $t>0$, consider the largest $s+1\leq t$ at
which $\sigma$ was accessed during the construction. 
If such stage does not exist, then $\hat{N}_t(\sigma\ast 0)=\hat{N}_t(\sigma\ast 1)=0$ and
\eqref{jEIyglTt} holds by the monotonicity of $f$ and its approximations $f_s$. 
Otherwise, according to the construction, and in particular \eqref{Qd5AiJd}, 
we must have 
\begin{equation}\label{ppv3c3rS}
\hat{N}_s(\sigma)\geq \hat{M}_s(\sigma)+\sum_{i\leq |\sigma|} 2^{-g(i)-1}
\hspace{0.3cm}
\textrm{and}
\hspace{0.3cm}
\hat{N}_{s+1}(\sigma)\geq \hat{M}_{s+1}(\sigma)+\sum_{i\leq |\sigma|} 2^{-g(i)-1}
\end{equation}
where the second inequality holds because 
$\hat{N}_{s+1}(\sigma)=\hat{N}_s(\sigma)$ (since $\sigma$ was
accessed at $s+1$ and not any of its predecessors) and 
$\hat{M}_{s+1}(\sigma)=\hat{M}_s(\sigma)$
(because otherwise a predecessor of $\sigma$ would 
have been accessed at stage $s+1$, or else 
$\sigma=\lambda$  and $\sigma$ would not have been accessed at stage $s+1$).
Moreover by the choice of $s$ we have 
$\hat{N}_{s+1}(\sigma\ast 0)=\hat{N}_t(\sigma\ast 0), 
\hat{N}_{s+1}(\sigma\ast 1)=\hat{N}_t(\sigma\ast 1)$. 
Hence
\begin{equation}
N_t(\sigma \ast 0) +N_t(\sigma \ast 1)=
\hat{N}_{s+1}(\sigma \ast 0) +\hat{N}_{s+1}(\sigma \ast 1)+2\cdot f_t(|\sigma|+1).
\end{equation}
According to \eqref{Qd5AiJd} we have
\begin{equation}\label{TEfoXi7ueV}
N_{s+1}(\sigma \ast 0) +N_{s+1}(\sigma \ast 1)\leq
2\cdot \Big(\sum_{i\leq |\sigma|} 2^{-g(i)-1} \Big)
+ 2^{-g(|\sigma|+1)}\cdot \Big(t_{s+1}(\sigma)+q_{s+1}(\sigma)\Big)
+2\cdot f_{s+1}(|\sigma|+1).
\end{equation}
By \eqref{an8o7LvU} and the fact that $M_{s+1}$ is a supermartingale we have:
\[
2^{-g(|\sigma|+1)}\cdot \Big(t_{s+1}(\sigma)+q_{s+1}(\sigma)\Big)+2\cdot f_{s+1}(|\sigma|+1)=  
\hat{M}_{s+1}(\sigma\ast 0)+\hat{M}_{s+1}(\sigma\ast 1)+2\cdot f_{s+1}(|\sigma|+1)\leq 
2\cdot \hat{M}_{s+1}(\sigma)+2\cdot f_{s+1}(|\sigma|)
\]
so plugging this back to \eqref{TEfoXi7ueV} we get
\[
N_{s+1}(\sigma \ast 0) +N_{s+1}(\sigma \ast 1)\leq 
2\cdot \Big(\sum_{i\leq |\sigma|} 2^{-g(i)-1} \Big)
+ 2\cdot \hat{M}_{s+1}(\sigma) 
+2\cdot f_{s+1}(|\sigma|).
\]
Then applying the second inequality of \eqref{ppv3c3rS} to the preceding inequality, we get
\[
N_{s+1}(\sigma \ast 0) +N_{s+1}(\sigma \ast 1)\leq 
2\cdot \big(\hat{N}_{s+1}(\sigma) + f_{s+1}(|\sigma|)\big)= 2N_{s+1}(\sigma).
\]
Since $\hat{N}_{s+1}(\sigma\ast 0)=\hat{N}_t(\sigma\ast 0), 
\hat{N}_{s+1}(\sigma\ast 1)=\hat{N}_t(\sigma\ast 1)$ and since $f_t(|\sigma|)-f_{s+1}(|\sigma|)\geq f_t(|\sigma|+1)-f_{s+1}(|\sigma|+1)$ this gives: 
\[
N_t(\sigma \ast 0) +N_t(\sigma \ast 1)\leq 
2\cdot N_{t}(\sigma).
\]
Hence for each $t$ the function $N_t$ is a computable supermartingale. 
By \eqref{LjU5qbvX} and the fact that $(f_s)$ is a \lce approximation to $f$,
it follows that $(N_s)$ is a \lce approximation to $N$. Hence $N$ is a \lce supermartingale.
In order to establish that $N$ is a $(g+1)$-granular supermartingale, recall that
$M$ is a $g$-granular supermartingale, and for each $\sigma$
the integer parameters $t_s(\sigma), q_s(\sigma)$ are nondecreasing and reach a limit
after finitely many stages.
Note that the only redefinition of $\hat{N}_{s+1}$ in the construction
occurs through \eqref{Qd5AiJd}. This, and the fact that
$t_s(\sigma),q_s(\sigma)$ are integers, shows that
each $\hat{N}_t$ is a $(g+1)$-granular function.
Hence the limit $\hat{N}$ of $(\hat{N}_s)$ is also
$(g+1)$-granular.
Then by Definition \ref{tkUAG5IyAt} it follows that
$N$ is a $(g+1)$-granular \ce supermartingale.
 
It remains to show that there exists a real $X$ such that $N(X\restr_n)\geq \sum_{ i\leq n} 2^{-g(i)-1}$
and $M(X\restr_n)\leq M(\lambda)$ for all $n\geq 0$. By the definition of $M,N$ and the fact that $f$ is
non-negative and nonincreasing, it suffices to show that
there exists a real $X$ such that $\hat{N}(X\restr_n)\geq \sum_{ i\leq n} 2^{-g(i)-1}$
and $\hat{M}(X\restr_n)\leq \hat{M}(\lambda)$ for all $n\geq 0$. 
The idea is as we described it at the beginning of the proof, \ie to let $X$ follow the path
where $M$ does not increase its capital.
Define $X$ inductively as follows. Given $X\restr_n$ define: 
\[
\mathbin{X}(n)=
\begin{cases}
0&\textrm{if $\hat{M}(X\restr_{n}\ast 0)\leq \hat{M}(X\restr_{n}\ast 1)$}\\
1&\textrm{if $\hat{M}(X\restr_{n}\ast 0)> \hat{M}(X\restr_{n}\ast 1)$.}
\end{cases}
\]
We shall establish the stronger condition that: 
\begin{equation}\label{fq4zOJyo}
\hat{N}(X\restr_n)\geq \sum_{i\leq n} 2^{-g(i)-1}+\hat{M}(X\restr_n)
\hspace{0.5cm}\textrm{for all $n$.}
\end{equation}
We prove this by induction on $n$.
It is clear that the claim holds for $n=0$.
Let $T_m=\lim_s t_s(X\restr_m)$ and
$Q_m=\lim_s q_s(X\restr_m)$ for each $m$, and suppose that 
\eqref{fq4zOJyo} holds for $n$. Suppose first that $T_n < Q_n$, 
so that $\hat{M}$ may be thought of as betting that $X(n)=1$, 
while $\hat{N}$ guesses correctly that $X(n)=0$. In this case it follows from the fact that $\hat{M}$ is 
a $g$-granular function that $Q_n\cdot 2^{-g(n+1)}\geq \hat{M}(X\restr_{n+1})+2^{-g(n+1)}$. 
From \eqref{Qd5AiJd} we then have 
\[
\hat{N}(X\restr_{n+1}) > \hat{M}(X\restr_{n+1}) +\sum_{i\leq n+1} 2^{-g(i)-1}.
\] 
Suppose next that $T_n = Q_n$. In this case we still have $X(n)=0$, but now 
$Q_n\cdot 2^{-g(n+1)}=\hat{M}(X\restr_{n+1})$. 
The final term $2^{-g(n+1)-1}$ in \eqref{Qd5AiJd}, however, 
means that  \eqref{fq4zOJyo} still  holds. Suppose finally that $T_n>Q_n$, so that $X(n)=1$. Then 
$T_n\cdot 2^{-g(n+1)}\geq\hat{M}(X\restr_{n+1})+2^{-g(n+1)}$, 
so then even though we subtract $2^{-g(n+1)-1}$ in  \eqref{Qd5AiJd}, 
we may again conclude that \eqref{fq4zOJyo} holds. 
 This completes the inductive step and the proof
of \eqref{fq4zOJyo}. 
Finally, it is clear from the above argument that if $f$ is constantly zero then
$N$ is a strongly $g$-granular \ce supermartingale. This shows the latter clause of the lemma.
\end{proof}

We make three observations regarding Lemma \ref{U26qaoGM}, which follow from its proof.
First, not only does $N$ succeed on $X$, but it does so in an essentially monotonic fashion, in the sense of
\eqref{fq4zOJyo}. Second, $N$ is obtained uniformly from $M$, in the sense that
there is a computable function which, given a \ce index for $M$ (\ie a program
which produces \lce approximations to $\hat{M}$ and $f$), produces a \ce index for $N$
with the prescribed properties.
Finally, the real $X$ is computable from $M$, which is a \lce function. Therefore $X$ is
computable from the halting problem.
Note that Lemma \ref{Bb4WohYKDQ}
is a special case of Lemma \ref{U26qaoGM} when $\sum_i 2^{-g(i)}=\infty$.

\subsection{Granular supermartingales and effective randomness}\label{5Hnk5ujNAa}
In this section we give a proof of Theorem \ref{Wab5YAaQ7s}.
For clause (b) of Theorem \ref{Wab5YAaQ7s}, suppose that we are given 
$g$ with the assumed properties. Consider
the universal \ce supermartingale $N$. By Lemma \ref{Bb4WohYKDQ},
given any $g$-granular supermartingale $M$ we can find $X$  such that
$\limsup_n M(X\restr_n)$ is finite while
$\limsup_n N_{\ast}(X\restr_n)$ is infinite for some \ce supermartingale $N_{\ast}$.
By the universality of $N$, the latter condition implies that
$\limsup_n N(X\restr_n)$ is also infinite, which concludes the proof of clause (b).
For clause (a), let $N$ be a  \ce supermartingale.
Given positive rational numbers $q,p$ let $\mathcal S(q,p)$ be the
largest multiple of $p$ which is less than  $q$.
For each string $\sigma$ we define 
\[
M(\sigma)=\sum_{i>|\sigma|} 2^{-g(i)} + \mathcal S(N(\sigma), 2^{-g(|\sigma|)})
\]
and note that $M$ is \ce as a function, because $N$ is a \ce function.
Moreover, $M$ is clearly $g$-granular,
and since $N$ is a supermartingale we have
\[
M(\sigma\ast 0)+M(\sigma\ast 1)\leq 
N(\sigma\ast 0)+N(\sigma\ast 1)+2\cdot \sum_{i>|\sigma|+1} 2^{-g(i)} 
\leq 2\cdot \left(N(\sigma)+\sum_{i>|\sigma|+1} 2^{-g(i)}\right)
\]
But by the definition of $S$ we have $N(\sigma)\leq   \mathcal S\left(N(\sigma), 2^{-g(|\sigma|)}\right) 
+2^{-g(|\sigma|)}$
so
\[
N(\sigma)+\sum_{i>|\sigma|+1} 2^{-g(i)}\leq 
\mathcal  S\left(N(\sigma), 2^{-g(|\sigma|)}\right) +\sum_{i>|\sigma|} 2^{-g(i)}=M(\sigma).
\]
Hence we may conclude that $M(\sigma\ast 0)+M(\sigma\ast 1)\leq M(\sigma)$ for all $\sigma$, which means
that $M$ is a \ce $g$-granular supermartingale. Also note that
\[
\sum_{i>|\sigma|} 2^{-g(i)} + N(\sigma)
\leq M(\sigma) +2^{-g(|\sigma|)}
\hspace{0.5cm}
\textrm{and}
\hspace{0.5cm}
M(\sigma) \leq 
\sum_{i>|\sigma|} 2^{-g(i)} + N(\sigma).
\]
The first inequality shows that  if  $\limsup_s N(X\restr_n)=\infty$
for some $X$, then $\limsup_s N(X\restr_n)=\infty$.
The second inequality above shows that  
if $\limsup_s M(X\restr_n)=\infty$ for some $X$ 
then $\limsup_s N(X\restr_n)=\infty$, which concludes
the proof of Theorem \ref{Wab5YAaQ7s}.

\section{Lower bounds on the redundancy in computation from random reals}\label{V1dKIq77PY}
In this section we give proofs of Theorems \ref{tbG7BLsZ}, \ref{cdUMjzgrN} and 
\ref{2NY34dfIMG}, using the result we now have for  restricted
betting strategies.
We start with the definition of redundancy, following G\'{a}cs \cite{MR859105}.

\begin{defi}[Oblivious use-function and redundancy]
We say that $f$ is a use-function of the Turing functional $\Phi$ if for every $X$ and $n$,
during the computation $\Phi^X(n)$ (whether it halts or not) all bits of $X$ that are
queried are smaller than $f(n)$. In this case we say that $\max\{f(n)-n,0\}$ is a redundancy of $\Phi$.
\end{defi}
Note that this definition is oblivious to the oracle $X$, a choice which reflects the fact that
we are interested in general upper bounds for the \KG theorem.
Clearly, given a Turing functional, there are many choices for its use function and its
redundancy. However we are generally interested in minimising the use-function and
the redundancy of computations. Moreover, we only consider use-functions and redundancy functions
which are computable and nondecreasing.
Given a Turing functional $\Phi$ with nondecreasing computable use-function $f$, we may view $\Phi$ as a partial computable function
which maps strings of length $f(n)$ to strings of length $n$ (for each $n$).
The following fact links Turing reductions with supermartingales.

\begin{lem}[Supermartingales from Turing functionals]\label{znvNmmPcdY}
Let $\Phi$ be a Turing functional with computable nondecreasing redundancy $g$,
and for each string $\nu$ let $h(\nu)$ be the number of strings $\tau$ of length
$|\nu|+g(|\nu|)$ such that $\Phi^{\tau}=\nu$.
Then the function  $h^{\ast}(\nu):=2^{-g(|\nu|)}\cdot h(\nu)$
is a strongly $g$-granular \ce supermartingale.
\end{lem}
\begin{proof}
Since $\Phi$ is a Turing functional,  we have
$h(\nu0)+h(\nu1)\leq 2^{|\nu|+1+g(|\nu|+1)-|\nu|-g(|\nu|)}\cdot h(\nu)$ and so
\[
h(\nu0)+h(\nu1)\leq 2^{1+g(|\nu|+1)-g(|\nu|)}\cdot h(\nu)
\hspace{0.6cm}
\textrm{for all strings $\nu$.}
\]
Since $h$ is an integer-valued function,
$h^{\ast}$ is a $g$-granular function. Moreover:
\[
h^{\ast}(\nu0)+h^{\ast}(\nu1)=
2^{-g(|\nu|+1)}\cdot \big(h(\nu0)+h(\nu1)\big)\leq
2^{-g(|\nu|+1)}\cdot 2^{1+g(|\nu|+1)-g(|\nu|)}\cdot h(\nu)=
2\cdot h^{\ast}_e(\nu).
\]
So $h^{\ast}$ is a strogly $g$-granular supermartingale.
Finally $h$ is a \lce function, because $\Phi$ is a Turing functional. So $h^{\ast}$ is a \lce
function, which concludes the proof.
\end{proof}

Lemma \ref{znvNmmPcdY} establishes a method for  constructing 
supermartingales from Turing reductions. Restricted wagers in the supermartingales constructed 
correspond to upper bounds on the oracle-use of the Turing reductions they are built from.
Similar arguments have been used in
\cite{MR2286414} and  \cite[Theorem 9.13.2]{rodenisbook},  for the special case of
integer-valued martingales and Turing reductions with constant redundancy. The underlying
general topic here is the connection between Turing functionals and semi-measures, which 
was explored in \cite{MR0307889}. For a recent account of this topic the reader is 
referred to \cite{BHPSr}, while \cite{rodenisbook} also contains related material
in various sections of Chapters 3,6 and 7.

We are now ready to apply 
Lemma \ref{U26qaoGM} in order to prove a density lemma,  which will be
the basis of an inductive construction specifying the reals required by
Theorem \ref{tbG7BLsZ}.   The proof of Theorem \ref{cdUMjzgrN} also establishes
Theorem \ref{tbG7BLsZ}, but is slightly more involved than a direct proof of the latter.
 We therefore choose to give a simple proof of 
Theorem \ref{tbG7BLsZ} in Section \ref{RW7DicsV4p}, before expanding that proof 
in order to obtain a proof of Theorem \ref{cdUMjzgrN}.

\subsection{Proof of Theorem \ref{tbG7BLsZ}}\label{RW7DicsV4p}
We use an effective forcing or finite extension argument, based on the following fact.
\begin{lem}[Density lemma]\label{XYRkqrIUo}
Let $\Phi$ be a Turing functional with redundancy a computable nondecreasing
function $g$ such that $\sum_i 2^{-g(i)}=\infty$.
Given any  $c\in \mathbb{N}$ and any finite string $\nu_0$, 
there exists an extension $\nu\supset\nu_0$ such that $K(\mu)< |\mu|-c$
for every string $\mu$ of length $|\nu|+g(|\nu|)$ for which $\Phi^{\mu}=\nu$.

\end{lem}
\begin{proof}
Given the functional $\Phi$, recall the associated functions $h, h^{\ast}$ from 
Lemma \ref{znvNmmPcdY}. 
Since $\Phi$ has redundancy $g$, it follows that 
$h^{\ast}$ is a $g$-granular \ce supermartingale.
Then given $\nu_0$, by
Lemma \ref{U26qaoGM} it follows that
 there exists a constant $d$, a \ce supermartingale $N$ and a real $Z\supset \nu_0$, such that
$h^{\ast}(Z\restr_n)<2^d$ for all $n$ and $N$ succeeds on $Z$ . By the
characterization of \ml randomness in terms of \ce supermartingales, 
it follows that
the real $Z$  
is not \ml random. Moreover
since $h^{\ast}(Z\restr_n)<2^d$ for all $n$, given the definition of $h^{\ast}$ we have
that $h(Z\restr_n)< 2^{d+g(n)}$ for all $n$.
We claim that  there exists a \pf machine $M$
such that:
\begin{equation}\label{1getCpYlbn}
\forall n\in\Nat, \mu\in 2^{n+g(n)}\hspace{0.3cm}\left(\Phi^{\mu}=Z\restr_n
\hspace{0.3cm}
\Rightarrow 
\hspace{0.3cm}
K_M(\mu)\leq K(Z\restr_n)+g(n)+d\right).
\end{equation}
The machine $M$ is defined in the following self-delimiting way.
Given a program $\sigma$, $M$ first looks for an initial segment $\sigma_0$ of $\sigma$
which is in the domain of the universal \pf machine $U$. If and when it finds $\sigma_0$, 
$M$ calculates $\tau=U(\sigma_0)$ -- one can think of the machine as  interpreting this string $\tau$ as
 $Z\restr_n$. It then calculates $g(n)$ (where $n$ is the length of $\tau$)
and reads $\sigma_1$, which is the following $g(n)+d$ bits of $\sigma$ (starting from bit $|\sigma_0|+1$).
If $\sigma$ does not have sufficiently many bits that $\sigma_1$ is defined then $M$ loops indefinitely. Otherwise, $M$ interprets
the string $\sigma_1$  as a number $t\leq 2^{d+g(n)}$.
It then interprets the number $t$ as the priority index of a string $\mu$
in the universal enumeration of strings $\rho$ such that
$\Phi^{\rho}=\tau$. In other words, $M$ runs this universal enumeration 
and starts producing the computably enumerable sequence of strings $\rho$ with
$\Phi^{\rho}=\tau$, stopping at the
$t$th such string $\mu$. If there are less than 
$t$ many strings $\rho$ such that $\Phi^{\rho}=\tau$, then $M$ loops indefinitely.
Finally $M$ assigns $\sigma_0\ast \sigma_1$ as a description of $\mu$ (\ie we define
$M(\sigma_0\ast \sigma_1)=\mu$). Since
$U$ is \pf and the length of $\sigma_1$ is determined by $\sigma_0$, the machine $M$
is \pfn.
Moreover, given a real $Z$ such that  
$h(Z\restr_s)< 2^{d+g(s)}$ for all $s$, $M$ will describe every string $\mu$
such that $\Phi^{\mu}=Z\restr_n$
with a string of length $K(Z\restr_n)+g(n)+d$. Indeed, by the property
$h(Z\restr_n)< 2^{d+g(n)}$, if $M$ is given as an input the concatenation of
a description of $Z_n$ and a string of length $d+g(n)$ which codes the priority index of string 
$\mu$ in the enumeration of all strings $\rho$ with  
$\Phi^{\rho}=Z\restr_n$, it will follow the steps above, and will eventually output the string $\mu$.
This completes the proof of \eqref{1getCpYlbn}.

We can now use our assumption that $Z$ is not \ml random in order to complete the proof of the lemma.
Since $M$ is a \pfm there exists some constant $c_0$ such that 
$K(\rho)<K_M(\rho)+c_0$ for all strings $\rho$.
Since $Z$ is not \ml random we can choose some $n>|\nu_0|$ such that $K(Z\restr_n)<n-c_0-d-c$.
Then given any string $\mu$ of  length $n+g(n)$  such that $\Phi_e^{\mu}=Z\restr_n$,
according to \eqref{1getCpYlbn} 
we have:
\[
K(\mu)< K(Z\restr_n)+g(n)+d+c_0 < n -d-c_0-c +g(n)+d+c_0 =n+g(n)-c=|\mu|-c.
\]
This concludes the proof of the lemma.
\end{proof}

Let $(\Phi_e, g_e)$ an effective enumeration of all pairs of Turing functionals $\Phi$ and
partial computable nondecreasing functions $g$ which are a redundancy function for $\Phi$.
This means that for each $e,n$ and each oracle $X$, 
if $\Phi^X_e(n)$ is defined then $g_e(n)$ is defined and the oracle-use
in the computation $\Phi^X_e(n)$ is bounded above by $n+g_e(n)$.
Let $I$ contain the indices $e$ such that $g_e$ is a total function with
$\sum_i 2^{-g(i)}=\infty$. For each $e$, let
$S(e,c)$ be the set of strings $\nu$ with the property that for all $\mu$ of length
$|\nu|+g(|\nu|)$ such that $\Phi_e^{\mu}=\nu$ we have $K(\mu)<|\mu|-c$.
Then the sets $S(e,c)$ are $\Sigma^0_2$, uniformly in $e,c$.
By Lemma \ref{XYRkqrIUo}, for each $e\in I$ and all $c$ the set $S(e,c)$ is dense,
\ie every string has an extension in $S(e,c)$.
Therefore every weakly 2-generic real has a prefix in $S(e,c)$, for each $e\in I$ and each $c$.
Theorem \ref{tbG7BLsZ} follows directly from this fact.

\subsection{Essential part of the proof of Theorem \ref{cdUMjzgrN}}\label{q4yEFszH5d}
We denote Turing reducibility by $\leq_T$.
In this section we show how to effectivize the argument of Section \ref{RW7DicsV4p} in order
to obtain a set $X\leq_T \emptyset'$ with the properties of Theorem \ref{tbG7BLsZ}.
Then in Section \ref{9VbXi7tSP7} we use  standard computability-theoretic
apparatus in order to show that for any given 
set $A$ which is generalized non-low$_2$ there exists such an $X$ with $X\leq_T A$,  thus completing the proof of Theorem \ref{cdUMjzgrN}.

In the proof of Lemma \ref{ImUc5BvB} we will need to restrict the enumeration
of strings $\mu$ such that $\Phi^{\mu}=\nu$. 
The following definition
introduces some notation for imposing such restrictions. 

\begin{defi}[Restricted enumeration of $\Phi$]\label{UiWgIw6Zw6}
Let $\Phi$ be a Turing functional with redundancy $g$.
Given any any string $\nu$, let
$\mathcal{Q}_0(\nu)$ be the set of all
strings $\mu$ of length $|\nu|+g(|\nu|)$ such that
$\Phi^{\mu}=\nu$, and let $\mathcal{Q}_0(\nu)[s]$ be a computable enumeration of this set. 
For each $d\in \mathbb{N}$ define
\[
\mathcal{Q}(d,\nu)=
\begin{cases}
\mathcal{Q}_0(\nu),&\textrm{if $|\mathcal{Q}_0(\nu)|<2^{d+g(|\nu|)}$;}\\
\mathcal{Q}_0(\nu)[s(d,\nu)], &\textrm{otherwise}.
\end{cases}
\]
where $s(d,\nu)$ is the largest stage such that 
$|\mathcal{Q}_0(\nu)|[s]<2^{d+g(|\nu|)}$, in the case where
$|\mathcal{Q}_0(\nu)|[s]|\geq 2^{d+g(|\nu|)}$, and $s(d,\nu)$ is undefined otherwise.
\end{defi}

Note that $\mathcal{Q}(d,\nu)$ is  uniformly \ce in $\Phi,g,\nu,d$. Of course the definition of $\mathcal{Q}(d,\nu)$ also depends upon $\Phi$ and the redundancy $g$, but these inputs will always be clear from context and so we suppress them for the sake of tidy notation.   
The following lemma will also be used in the proof of Lemma 
\ref{ImUc5BvB}.

\begin{lem}[Turing functionals and \pf complexity]\label{oAGsjj2ali}
Let $\Phi$ be a Turing functional with redundancy $g$.
There exists a \pf machine $M$
such that 
\begin{equation*}
\forall d,\nu\ \ \forall \mu\in \mathcal{Q}(d,\nu)\hspace{0.3cm}
\Big(\Phi^{\mu}=\nu
\hspace{0.3cm}\Rightarrow \hspace{0.3cm}
K_M(\mu)\leq K(\nu)+g(|\nu|)+d\Big).
\end{equation*}
Moreover an index for $M$ is uniformly computable from indices for $\Phi,g$.
\end{lem}
\begin{proof}
Such a machine $M$ can be constructed  as in the proof of Lemma \ref{XYRkqrIUo}.  
\end{proof}

\begin{lem}[Effective density lemma]\label{ImUc5BvB}
Let $\Phi$ be a Turing functional  with computable nondecreasing redundancy $g$.
There exists a computable function $f$  such that for every $c,\nu_0$:
\begin{equation}\label{99fMzMwUT5}
\sum_i 2^{-g(i)}>f(c,\nu_0)\Rightarrow
\Big[\exists \nu\supset\nu_0\ \forall\mu\in 2^{|\nu|+g(|\nu|)}\ 
\Big(\Phi^{\mu}=\nu\Rightarrow K(\mu)< |\mu|-c\Big)\Big].
\end{equation}
Moreover an index of $f$ can be obtained effectively from
 indices for $\Phi,g$.
\end{lem}
\begin{proof}
Given $\Phi,g$ as in the hypothesis, consider the functions $h, h^{\ast}$ of 
Lemma \ref{znvNmmPcdY}. By
Lemma \ref{U26qaoGM}, for each string $\nu_0$ 
there exists a constant $d=d(\nu_0)$ (an example of which we can find effectively since all that is required is an upper bound for $h^{\ast}(\nu_0)$) and there exist
a \lce supermartingale $N$ and a real $Z$ such that:
\[
\Big(Z\supset\nu_0\Big)
\hspace{0.2cm}\bigwedge \hspace{0.2cm}
\Big(\textrm{$h(Z\restr_n)< 2^{d(\nu_0)+g(n)}$ \hspace{0.1cm} for all $n$}\Big) 
\hspace{0.2cm}\bigwedge\hspace{0.2cm}
\Big(N(Z\restr_n)\geq \sum_{i=0}^n 2^{-g(i)}
\hspace{0.2cm}
\textrm{for all $n>|\nu_0|$}\Big).
\]
Note that the second clause of the conjunction above follows since 
$h^{\ast}(Z\restr_n)<2^d$ implies
$h(Z\restr_n)< 2^{d+g(n)}$.
Moreover, as observed in Section \ref{RBq4QNezFK},
an index for $N$ (together with an upper bound for $N(\lambda)$) can be obtained effectively from $\nu_0$ and indices for $\Phi$ and $g$.
So by Lemma \ref{QhFm7xcjbM}, there exists a computable function
$f_0$ (whose index is computable from the indices of $\Phi,g$)
such that for all $t\in\Nat$ we have:
\begin{equation}\label{jLEez1kHp3}
\sum_i 2^{-g(i)}>f_0(c,\nu_0)\Rightarrow
\exists \nu\supset\nu_0\ \Big(h(\nu)<2^{d(\nu_0)+g(|\nu|)}\ \wedge\ K(\nu)<|\nu|-c\Big).
\end{equation}
Now consider the machine $M$ of Lemma \ref{oAGsjj2ali}, and let $m$ be its
index, which is a computable function of indices for  $\Phi,g$.
We define $f(c,\nu_0)=f_0(c+m+1+d(\nu_0),\nu_0)$ for each $c,\nu_0$, 
and show that $f$ meets 
condition \eqref{99fMzMwUT5}.
Given our choice for the 
underlying universal \pf machine $U$,  we have
$K(\rho)<K_M(\rho)+m+1$ for all strings $\rho$.
Fix $c,\nu_0$ and assume that the left-hand-side of \eqref{99fMzMwUT5} holds.
Then by \eqref{jLEez1kHp3} and the definition of $f$, there
exists an extension $\nu$ of $\nu_0$
such that $K(\nu)<|\nu|-c-d(\nu_0)-m-1$ and $h(\nu)<2^{d(\nu_0)+g(|\nu|)}$.
By Definition \ref{UiWgIw6Zw6} the latter inequality implies that
\[
\big\{\mu\in 2^{|\nu|+g(|\nu|)}\ |\ \Phi^{\mu}=\nu\big\}=\mathcal{Q}(d(\nu_0),\nu).
\]
From
Lemma \ref{oAGsjj2ali} it follows  
that for all strings $\mu$ of length $|\nu|+g(|\nu|)$ with $\Phi^{\mu}=\nu$:
\[
K_M(\mu)\leq 
K(\nu)+g(|\nu|)+d(\nu_0)\leq 
\big(|\nu|-c-m-1-d(\nu_0)\big) +g(|\nu|)+d(\nu_0). 
\]
This establishes that
$K(\mu)\leq |\nu|+g(|\nu|)-c=|\mu|-c$.
Finally observe that $f$ is obtained effectively from indices for $\Phi$ and $g$, which
concludes the proof of the lemma.
\end{proof}

From the above proof and 
Lemma \ref{QhFm7xcjbM} we can see that
Lemma \ref{ImUc5BvB} also holds for partial computable functions $g$
in the following sense. 
Let $(\Phi_e, g_e)$ an effective enumeration of all pairs of Turing functionals $\Phi$ and
partial computable nondecreasing functions $g$.
\begin{equation}\label{2LHcslocWH}
\parbox{14cm}{There exists a computable function $f_{\ast}$ such that for each $e,c,\nu_0,k$, if
$g_e(i)\de$ for all $i\leq k$ and $\sum_{i=0}^k 2^{-g_e(i)}>f_{\ast}(e,c,\nu_0)$
then there exists an extension $\nu$ of $\nu_0$ of length $k$ such that
$K(\mu)< |\mu|-c$ for all strings $\mu$ of length
$|\nu|+g_e(|\nu|)$ such that $\Phi_e^{\mu}=\nu$.}
\end{equation}

Note that $f$ of Lemma \ref{ImUc5BvB} had two arguments, while $f_{\ast}$ has three arguments,
as it deals with every potential redundancy function $g_e$.
We are now ready to prove Theorem \ref{cdUMjzgrN}.
Given Lemma \ref{ImUc5BvB}, this is a standard argument in computability theory.
We first describe the construction of a real $X\leq_T \emptyset'$ 
which meets the requirements
of the theorem, which are:
\[
\mathcal{R}_{e,c}:\ 
\textrm{If $g_e$ is total and \ \ $\sum_i 2^{-g_e(i)}=\infty$\ \ \  \ then\ \ \ }
\forall Y\ \  \Big(\Phi_e^Y=X\Rightarrow\exists n\ \ K(Y\restr_n)\leq n-c
\Big).
\]
The proofs of the full claims of  
Theorems \ref{cdUMjzgrN} and \ref{2NY34dfIMG}, 
regarding generalized non-low$_2$ and array noncomputable sets, will be
modifications of this simpler case.
Let $\tuple{.,.}:\Nat\times\Nat\to\Nat$ be a computable bijection.
We describe a finite extension construction for the set $X$, which is computable from
the halting problem. The main issue here is that the left-hand-side of the 
outer implication in $\mathcal{R}_{e,c}$ is not computable from the halting problem.
This is the reason why we need Lemma \ref{ImUc5BvB} and \eqref{2LHcslocWH},
and not just Lemma \ref{XYRkqrIUo}.

We define a monotone sequence of strings $(\sigma_i)$, beginning with the empty string
$\sigma_0$ and eventually defining $X=\cup_i\sigma_i$. At stage $\tuple{e,c}+1$ we meet
$\mathcal{R}_{e,c}$. Let $f_{\ast}$ be the function from \eqref{2LHcslocWH}.
Inductively assume that $\sigma_i$, $i\leq \tuple{e,c}$ have been defined. At stage $\tuple{e,c}+1$
we 
ask if there exists $k$ such that
\[
\textrm{$g_e(i)\de$\ \  for all $i\leq k$
\hspace{0.4cm} and \hspace{0.4cm} 
$\sum_{i=0}^k 2^{-g_e(i)}>f_{\ast}(e,c,\nu_0)$.}
\]
 If not, we simply let $\sigma_{\tuple{e,c}+1}$ be
$\sigma_{\tuple{e,c}}\ast 0$. 
Otherwise we search for a proper extension $\nu$ of $\sigma_{\tuple{e,c}}$
of length at most $k$ such that
\begin{equation}\label{JJwTv34f2O}
\forall\mu\in 2^{|\nu|+g_e(|\nu|)}\ \ \ 
\Big(\Phi_e^{\mu}=\nu\Rightarrow K(\mu)< |\mu|-c\Big).
\end{equation}

By \eqref{2LHcslocWH} such a string $\nu$ exists. In this case we 
define $\sigma_{\tuple{e,c}+1}=\nu$. This completes the inductive definition of $(\sigma_i)$
and $X$.
An inspection of the construction
suffices to verify that only $\Sigma^0_1$ questions are asked, so $X\leq_T \emptyset'$.
Moreover for each $e,c$, condition $\mathcal{R}_{e,c}$ is met by all reals
extending $\sigma_0\ast\cdots\ast\sigma_{\tuple{e,c}+1}$. So 
the real $X=\cup_i\sigma_i$ meets $\mathcal{R}_{e,c}$ for all $e,c$.

\subsection{Proof of Theorem \ref{cdUMjzgrN}}\label{9VbXi7tSP7}
Suppose that $A$ is a generalized non-low$_2$ set.
We modify the construction of the previous section so as to build $X\leq_T A$.
The requirements to be satisfied are  
$\mathcal{R}_{e,c}$ as before. 
Recall that since $A$ is a generalized non-low$_2$,
for every $\Delta^0_2$ function $n\mapsto p(n)$ 
there exists a function 
$n\mapsto q(n)$ which is computable from $A$ and is not dominated
by $p$, i.e.\ such that there exist infinitely many $n$ for which $q(n)>p(n)$. 
The rough idea is the same as always when modifying constructions with oracle the halting set, in order to work below $A$ which is generalized non-low$_2$. One defines a function $p$ which is computable in the halting set, and which gives an upper bound for the length of computable search required at each stage of the construction in order to proceed `correctly'.  Then one shows that, in fact, it suffices to use $q$ which is not dominated by $p$ in order to bound the search at each stage.

\subsubsection{Dominating function and dynamics of strategies}\label{3RCakbwU1}
We first define a function $p\leq_{T} \emptyset'$ which is sufficiently fast growing
so that it provides good approximations to the conditions involved in 
$\mathcal{R}_{e,c}$.
Recall the definition of $f_{\ast}$ from \eqref{2LHcslocWH}. We may assume that $f_{\ast}(e,c,\nu_0)>|\nu_0|$. 

We let $p_0(e,c,\nu_0)$ be the least $k>|\nu_0|$ such that 
\begin{equation}\label{iFNsv7jHdK}
\textrm{$g_e(i)\de$ for all $i\leq k$ 
\hspace{0.5cm}and\hspace{0.5cm}  
$\sum_{i=0}^k 2^{-g_e(i)}>f_{\ast}(e,c,\nu_0)$}
\end{equation}
if there exists such, and we define $p_0(e,c,\nu_0)=0$ otherwise. Let $p_0(e,c,\nu_0)=k$. Then we define $p_1(e,c,\nu_0)$ to be the least $s>k$ such that: 
\begin{enumerate}[\hspace{0.5cm}(a)]
\item $g_e(i)[s]\downarrow$ for all $i\leq k$ if $k>0$. 
\item For all $\nu$ of length at most $k$ and $\mu$ of   length $|\nu|+g_e(|\nu|)$ such that 
$\Phi_e^{\mu}\downarrow =\nu$, we have  $\Phi_e^{\mu}[s]\downarrow =\nu$. 
 \item For all $\mu$ of length at most $k+g(k)$,  $K(\mu)$ has settled by stage $s$, \ie $K_s(\mu)=K(\mu)$.
\end{enumerate}

Finally, we define $p(s)$ to be the least number greater than $p_1(e,c,\nu_0)$ for all $e,c$ such that $\tuple{e,c}\leq s$ and $\nu_0$ of length at most $s$.

Clearly $p\leq_{T}\emptyset'$. 
Now fix  a function $q\leq_{T} A$ which is not dominated by $p$.

We are going to use $q$ in order to
construct $X$ which meets all requirements $\mathcal{R}_{e,c}$.
This will also be a finite extension construction, but it is important to ensure that
the length of $X$ that has been determined at stage $e$ is of length $e$.
This ensures that when we encounter some $e$ such that 
$q(e)\geq p(e)$, it will not be too late to make the right decision in terms of satisfying
some requirement of high priority that had remained unsatisfied in the previous stages.
At the start of each stage $s+1$ the initial segment $X\restr_s$ has been defined in the previous stages
and we are called to define $X(s)$, therefore specifying $X\restr_{s+1}$.

At stage $s+1$ we say that $\mathcal{R}_{e,c}$ {\em requires attention}
if it has not already been declared satisfied, and:
\begin{enumerate}[\hspace{0.5cm}(i)]
\item there exists a least $k>s$ such that: $k<q(s)$, for all $i\leq k$ we have $g_e(i)[q(s)]\downarrow$, and $\sum_{i=0}^k 2^{-g_e(i)} >f_{\ast}(e,c,X\restr_s)$;  
\item for this least $k$, there exists $\nu \supset X\restr_s$ of length $k$ which satisfies the following condition:  for all $\mu$ of length $k+g_e(k)$ such that $\Phi_e^{\mu}[q(s)]\downarrow =\nu$, $K_{q(s)}[\mu]<|\mu|-c$. 
\end{enumerate}

In this case we also say that $\mathcal{R}_{e,c}$ requires attention via $\nu$ for the lexicographically least $\nu$ satisfying the conditions of (ii) above.

\subsubsection{Construction of the real}\label{e7EemRzl2D}
At stage $0$: Define $X\restr_0 =\lambda$. 

At stage $s+1$: If there does not exist $\tuple{e,c}\leq s$ such that  $\mathcal{R}_{e,c}$ requires attention, then define $X(s)=0$. Otherwise, let  $\tuple{e,c}$ be that of highest priority, and let $\nu$ be such that  $\mathcal{R}_{e,c}$ requires attention via $\nu$. Define $X(s)=\nu(s)$. If $|\nu|=s+1$ then declare $\mathcal{R}_{e,c}$ to be satisfied. 

\subsubsection{Verification of the construction}\label{QUgXhTSL9}
Suppose that no requirement of higher priority than $\mathcal{R}_{e,c}$ requires attention at any stage $>s_0$. We show that $\mathcal{R}_{e,c}$ is satisfied, and that there exists a stage after which this requirement does not require attention. Let $s_1>s_0$ be such that $q(s_1)>p(s_1)$. From the definition of $p(s_1)$ and the fact that $q(s_1)>p(s_1)$ it follows that if $\mathcal{R}_{e,c}$ does not require attention at stage $s_1+1$ then either $g_e$ is not total, or else $g_e$ is total and $\sum_i 2^{-g_e(i)}$ is finite. If $g_e(i)\uparrow$ for some $i$,  then  $\mathcal{R}_{e,c}$ cannot require attention subsequent to stage $i$. If  $\sum_i 2^{-g_e(i)}$ is finite, then it follows directly from our assumption that $f_{\ast}(e,c,\nu_0)>|\nu_0|$ that $\mathcal{R}_{e,c}$ can only require attention at finitely many stages. So suppose, on the other hand, that $\mathcal{R}_{e,c}$ requires attention at stage $s_1+1$ via $\nu$. In this case the requirement will be declared satisfied by the end of stage $|\nu|$.

\subsection{Proof of Theorem \ref{2NY34dfIMG}}
Recall that given an array noncomputable set $A$ 
and any function $p$ which is 
weak truth-table computable in $\emptyset'$,
there exists a function $q\leq_T A$ which is not dominated
by $p$. So
if the function $p$ of Section \ref{3RCakbwU1} was computable from the halting problem with
{\em computable bound on the oracle use}, then we would have proved
Theorem \ref{cdUMjzgrN} under the weaker assumption of  array noncomputability.
Unfortunately, this is not the case. However by the same argument 
we can obtain a nonuniform version of
Theorem \ref{cdUMjzgrN}, under the weaker hypothesis of array noncomputability on the oracle $A$.
Let $A,g$ be as in the statement of Theorem \ref{2NY34dfIMG}, and let $e$ be an index of $g$.
We wish to construct $X\leq_T A$ which satisfies all requirements
$\mathcal{R}_{e,c}$, $c\in\Nat$ of Section \ref{q4yEFszH5d},  for the fixed index $e$ of $g$.
 The crucial point is that if we fix $e$ such that $g_e$ is total and follow the definition of 
 $p$ of Section \ref{3RCakbwU1} restricting to this fixed $e$, then the corresponding function $p_e$ is
  computable from the halting problem with
{\em computable bound on the oracle use}.
Hence, if $A$ is array noncomputable, we may choose an increasing function $q_e\leq A$ which is not
dominated by $p_e$.
Then the construction of Section \ref{e7EemRzl2D}, restricted to a fixed $e$ such that
$g_e=g$, gives a real $X\leq_T q_e$ which meets all requirements $\mathcal{R}_{e,c}, c\in\Nat$.
Again, the proof of this fact is the argument of Section \ref{QUgXhTSL9}, restricted to our fixed $e$.
Hence $X\leq_T  q_e\leq_T A$ and $X$ has the properties claimed in Theorem
\ref{2NY34dfIMG}.

\section{Conclusions}\label{TueOFoB6Hf}
Ku\v{c}era \cite{MR820784, Kucera:87}
and G{\'a}cs \cite{MR859105}
showed that every real is computable from a random real.
The best known general upper bound for the redundancy of such computations
is $\sqrt{n}\cdot \log n$, and is due to
G{\'a}cs \cite{MR859105} (Merkle and Mihailovi\'{c}
\cite{jsyml/MerkleM04} have provided a different proof of this fact).
In the present paper we asked for the optimal redundancy that can be achieved in the
\KG theorem. We showed that no computable nondecreasing function $g$
such that $\sum_i 2^{-g(i)}=\infty$ can be such an upper bound and demonstrated that
a large class of oracles require larger redundancy when they are computed by random reals.
This result improves the constant bound obtained
by Downey and Hirschfeldt \cite[Theorem 9.13.2]{rodenisbook}. 
Our result shows that, in general, the redundancy cannot be as slow growing as $\log n$, but
a large exponential gap with the currently known bound of $\sqrt{n}\cdot \log n$ remained. Recently
it was shown in \cite{optcod16} that the strict lower bounds that we obtain in the present paper are optimal.
In other words, any
computable nondecreasing function $g$
such that $\sum_i 2^{-g(i)}<\infty$ is a general upper bound on the redundancy in the computation of any real
from some \ml random oracle. This provides a complete characterization of the redundancy bounds in the \KG theorem.

\end{document}